\documentclass[a4paper]{article}
\usepackage{mathptmx}
\usepackage[utf8]{inputenc}
\usepackage[T1]{fontenc}
\usepackage[final]{microtype}
\usepackage[english]{babel}
\usepackage{csquotes}
\usepackage{amsmath}
\usepackage{amsthm}
\usepackage{amsfonts}
\usepackage{commath}
\usepackage{bm}
\usepackage{xfrac}
\usepackage{mathtools}
\usepackage{authblk}
\usepackage{url}
\usepackage[pdfencoding=auto,psdextra,hidelinks]{hyperref}
\usepackage{bookmark}
\usepackage{float}
\usepackage{booktabs}
\usepackage{subcaption}
\usepackage{tikz}
\usetikzlibrary{calc,positioning,shapes,arrows}

\newcommand*{\email}[1]{\href{mailto:#1}{\nolinkurl{#1}}}

\newcommand{\Abs}[1]{\ensuremath{\left| #1 \right|}}

\newcommand{\Simplex}{\ensuremath{\Delta}}
\newcommand{\CornerSimplex}{\ensuremath{\Delta_\mathrm{c}}}
\newcommand{\eps}{\ensuremath{\varepsilon}}

\newcommand{\RR}{\ensuremath{\mathbb{R}}}

\newcommand{\zo}{\ensuremath{\{0,1\}}}

\newcommand{\PTIME}{\ensuremath{\mathrm{P}}}
\newcommand{\NP}{\ensuremath{\mathrm{NP}}}
\newcommand{\coNP}{\ensuremath{\mathrm{coNP}}}
\newcommand{\coDP}{\ensuremath{\mathrm{coDP}}}
\newcommand{\PSPACE}{\ensuremath{\mathrm{PSPACE}}}

\newcommand{\cETR}{\ensuremath{\exists\RR}}
\newcommand{\coETR}{\ensuremath{\forall\RR}}

\newcommand{\PH}{\ensuremath{\mathrm{PH}}}
\newcommand{\DNP}{\ensuremath{\mathrm{DNP}}}
\newcommand{\coDNP}{\ensuremath{\mathrm{coDNP}}}
\newcommand{\DSigma}{\ensuremath{\mathrm{D\Sigma}}}
\newcommand{\DPi}{\ensuremath{\mathrm{D\Pi}}}
\newcommand{\Sigmap}{\ensuremath{\Sigma^\mathrm{p}}}
\newcommand{\Pip}{\ensuremath{\Pi^\mathrm{p}}}
\newcommand{\cForallExistsR}{\ensuremath{\forall\exists\RR}}
\newcommand{\cExistsForallR}{\ensuremath{\exists\forall\RR}}
\newcommand{\existsR}{\ensuremath{\exists^\RR}}
\newcommand{\existsD}{\ensuremath{\exists^\mathrm{D}}}
\newcommand{\forallR}{\ensuremath{\forall^\RR}}
\newcommand{\forallD}{\ensuremath{\forall^\mathrm{D}}}
\newcommand{\existsROP}{\ensuremath{\existsR\cdot}}
\newcommand{\existsDOP}{\ensuremath{\existsD\cdot}}
\newcommand{\forallROP}{\ensuremath{\forallR\cdot}}
\newcommand{\forallDOP}{\ensuremath{\forallD\cdot}}
\newcommand{\pairing}[2]{\ensuremath{\langle #1 , #2 \rangle}}
\newcommand{\poly}{\ensuremath{\mathrm{poly}}}

\newcommand{\theory}{\ensuremath{\mathrm{Th}}}
\newcommand{\Etheory}{\ensuremath{\theory_\exists}}

\newcommand{\QUAD}{\textsc{Quad}}
\newcommand{\FOURFEAS}{\textsc{4Feas}}

\newcommand{\HOMFOURFEAS}{\textsc{Hom4Feas}}
\newcommand{\PosSLP}{\textrm{PosSLP}}
\newcommand{\BP}{\ensuremath{\operatorname{BP}}}
\newcommand{\ExistsESS}{\ensuremath{\exists\mathrm{ESS}}}
\newcommand{\IsESS}{\ensuremath{\mathrm{IsESS}}}
\newcommand{\ExistsLSS}{\ensuremath{\exists\mathrm{LSS}}}
\newcommand{\IsLSS}{\ensuremath{\mathrm{IsLSS}}}

\newcommand{\calG}{\ensuremath{\mathcal{G}}}

\newcommand{\calC}{\ensuremath{\mathcal{C}}}
\newcommand{\support}{\ensuremath{\operatorname{Supp}}}
\newcommand{\Sym}{\ensuremath{\operatorname{Sym}}}

\newcommand{\Exp}{\operatorname*{E}}

\newtheorem{theorem}{Theorem}
\newtheorem{proposition}{Proposition}
\newtheorem{lemma}{Lemma}
\newtheorem{corollary}{Corollary}

\newtheorem{definition}{Definition}

\usepackage[
  backend=bibtex,
  style=alphabetic,
  sortlocale=en_US,
  useprefix=false,
  url=false,
  doi=true,
  eprint=false,
  maxbibnames=99,
]{biblatex}
\begin{document}
\title{Computational Complexity of Multi-Player Evolutionarily Stable Strategies\thanks{This paper appeared previously in a preliminary form~\cite{CSR:BlancH2021}}}
\author[1]{Manon Blanc}
\author[2]{Kristoffer Arnsfelt Hansen\thanks{Supported by the Independent Research Fund Denmark under grant no. 9040-00433B.}}
\affil[1]{ENS Paris-Saclay\\\email{manonblanc@free.fr}}
\affil[2]{Aarhus University\\\email{arnsfelt@cs.au.dk}}

\maketitle

\begin{abstract}
  In this paper we study the computational complexity of computing an
  evolutionary stable strategy (ESS) in multi-player symmetric
  games. For two-player games, deciding existence of an ESS is
  complete for $\Sigmap_2$, the second level of the polynomial time
  hierarchy.  We show that deciding existence of an ESS of a
  multi-player game is closely connected to the second level of the
  \emph{real} polynomial time hierarchy. Namely, we show that the
  problem is hard for a complexity class we denote as
  $\existsDOP \coETR$ and is a member of $\cExistsForallR$, where the
  former class restrict the latter by having the existentially
  quantified variables be Boolean rather then real-valued. As a
  special case of our results it follows that deciding whether a given
  strategy is an ESS is complete for $\coETR$.

  A concept strongly related to ESS is that of a locally superior
  strategy (LSS). We extend our results about ESS and show that
  deciding existence of an LSS of a multiplayer game is likewise hard
  for $\existsDOP \coETR$ and a member of $\cExistsForallR$, and as a
  special case that deciding whether a given strategy is an LSS is
  complete for $\coETR$.
\end{abstract}

\section{Introduction}
First introduced by Maynard Smith and Price in ~\cite{Nature:MaynardSmithP1973,JTB:MaynardSmith1974}, a central concept emerging
from evolutionary game theory is that of an evolutionary stable strategy (ESS)
in a symmetric  two-player game in strategic form.
Each pure strategy of the game is viewed as a type of possible individuals of a
population. A mixed strategy of the game then corresponds to
describing the proportion of each type of individual of the
population, which as a simplifying assumption is considered to be
infinite. The population is engaged in a pairwise conflict where two
individuals are selected at random and receive payoffs depending on
their respective types. The population is expected to evolve in a way
where strategies that achieve a higher payoff than others will spread
in the population. A strategy $\sigma$ is an ESS if it outperforms any
``mutant'' strategy $\tau\neq\sigma$ adopted by a small fraction of
the population. Otherwise we say that $\sigma$ may be invaded. An ESS
is in particular a symmetric Nash equilibrium (SNE), but, unlike a
SNE, it is not guaranteed to exist.

\begin{figure}[h]
  \centering
  \begin{tabular}{c|c|c|}
    \multicolumn{1}{c}{} &  \multicolumn{1}{c}{Hawk} & \multicolumn{1}{c}{Dove} \\\cline{2-3}
    Hawk & -1,-1 & 2,0  \\\cline{2-3}
    Dove & 0,2 & 1,1 \\\cline{2-3}
  \end{tabular}
  \caption{Hawk-Dove game}
  \label{FIG:HawkDove}
\end{figure}
The Hawk-Dove game~\cite{Nature:MaynardSmithP1973}, presented with
concrete payoffs in Fig.~\ref{FIG:HawkDove}, is a classic example
where an ESS may explain the proportion of the population tending to
engage in aggressive behavior.  The game has a unique SNE $\sigma$,
where the players choose Hawk with probability~$\frac{1}{2}$, and this
is in fact an ESS. Note first that
$u(\sigma,\sigma)= (-1)\left(\frac{1}{2} \right)^2 + 2\left(
  \frac{1}{2} \right)^2 + 0\left( \frac{1}{2} \right)^2 + 1\left(
  \frac{1}{2} \right)^2 = \frac{1}{2}$. Consider now any strategy
profile $\tau$ that chooses Hawk with probability~$p$. Then
$u(\tau,\sigma) = (-1 + 2) p/2 +(1+0)(1-p)/2 = \frac{1}{2}$ as
well. However, $u(\sigma,\tau)=\frac{3}{2}-2p$ and
$u(\tau,\tau)=1-2p^2$, and thus
$u(\sigma,\tau) - u(\tau,\tau) = 2(p-\frac{1}{2})^2$, which means that
$\sigma$ outperforms $\tau$ if $p \neq \frac{1}{2}$.

While the two-player setting is the typical setting to study ESS, the
concept may in a natural way be generalized to the setting of
multi-player games, as established by Palm~\cite{JMB:Palm1984} and
Broom, Cannings, and Vickers~\cite{BMB:BroomCV1997}. This allows one
to model populations that engage in conflicts involving more than two
individuals. Many of the two-player games typically studied in the
context of ESS readily generalize to multi-player games, including the
Hawk-Dove and Stag Hunt games (cf.\ \cite{Book:BroomRychtar2013}). For a naturally occurring example,
Broom and Rychtář~\cite[Example~9.1]{Book:BroomRychtar2013} argue that
the cooperative hunting method of carousel feeding by killer whales
may be modeled as a multi-player Stag Hunt game.

The computational complexity of computing an ESS was first studied by
Etessami and Lochbihler~\cite{IJGT:EtessamiL2007}. We shall denote the
problem of deciding whether a given symmetric game in strategic form
has an ESS as $\ExistsESS$ and similarly the problem of deciding
whether a given strategy is an ESS of the given game as
$\IsESS$. Previous work has been concerned only with two-player
symmetric games in strategic form. Etessami and Lochbihler proved that
$\ExistsESS$ is hard both for $\NP$ and $\coNP$ and is contained in
$\Sigmap_2$. Nisan~\cite{ECCC:Nisan2006} showed that $\ExistsESS$ is
hard for the class $\coDP$, which is the class of unions of languages
from $\NP$ and $\coNP$. From both works it also follows that the
problem $\IsESS$ is $\coNP$-complete. Finally
Conitzer~\cite{MOR:Conitzer2019} showed $\Sigmap_2$-completeness for
$\ExistsESS$. The direct but important consequence of these results is
that any algorithm for computing an ESS in a general game can be used
to solve $\Sigmap_2$-complete problems. For instance, we cannot expect
to be able to compute an ESS in a simple way using a SAT solver.

One may observe that the above hardness results for two-player games
also generalize to apply to $m$-player games, for any fixed $m\geq
3$. Note that, since a reduction showing $\Sigmap_2$-hardness must
produce an $m$-player symmetric game, this is not a trivial
observation (in particular adding ``dummy'' players, each having a
single strategy, to a nontrivial symmetric game would result in a
non-symmetric game). One would however suspect that the problems
$\ExistsESS$ and $\IsESS$ become significantly harder for $m$-player
games, when $m \geq 3$. Namely, starting with the work of Schaefer and
Štefankovič~\cite{TOCS:SchaeferS2017}, several works have shown that
many natural decision problems concerning Nash equilibrium (NE) in
3-player strategic form games are
$\cETR$-complete~\cite{TEAC:GargMVY2018,STACS:BiloM2016,STACS:BiloM17,TCS:Hansen2019,SAGT:BerthelsenH19}. These
results stand in contrast to the two-player setting, where the same
decision problems are
$\NP$-complete~\cite{GEB:GilboaZemel89,GEB:ConitzerS08}. The class
$\cETR$ is the complexity class that captures the decision problem for
the existential theory of the reals~\cite{TOCS:SchaeferS2017}, or
alternatively, is the constant-free Boolean part of the real analogue
$\NP_\RR$ in the Blum-Shub-Smale model of
computation~\cite{FCM:BurgisserC2009}. Clearly we have
$\NP \subseteq \cETR$, and from the decision procedure for the
existential theory of the reals by Canny~\cite{STOC:Canny1988} it
follows that $\cETR \subseteq \PSPACE$. We consider it likely that
$\NP$ is a strict subset of $\cETR$, which would mean that the above
mentioned decision problems concerning NE become strictly harder as
the number of players increase beyond two.

We confirm that the problems $\ExistsESS$ and $\IsESS$ indeed are
likely to become harder for multi-player games by proving hardness of
the problems for discrete complexity classes defined in terms of real
complexity classes that we consider likely to be stronger than
$\Sigmap_2$ and $\NP$.  Our results are perhaps most easily stated in
terms of the decision problem for the first order theory of the reals
$\theory(\RR)$. Just like the class $\cETR$ corresponds to the existential
fragment $\theory_\exists(\RR)$ of $\theory(\RR)$, we can consider
classes $\coETR$ and $\cExistsForallR$ corresponding to the universal
fragment $\theory_\forall(\RR)$ and the existential-universal fragment
$\theory_{\exists\forall}(\RR)$ of $\theory(\RR)$, respectively. It is
easy to see that the problem $\ExistsESS$ belongs to $\cExistsForallR$
and that $\IsESS$ belongs to $\coETR$. We show that for 5-player
games, the problem $\ExistsESS$ is hard for the subclass of
$\cExistsForallR$ where the block of universal quantifiers is
restricted to range over Boolean variables. For the problem $\IsESS$
we completely characterize its complexity for 5-player games by
proving that the problem is also hard for $\coETR$. Our hardness
results thus imply that any algorithm for computing an ESS in a
5-player game can be used to solve quite general problems involving
real polynomials. In particular it indicates that computing an ESS is
significantly more difficult than deciding if a system of real
polynomials has no solution, which is a basic problem complete for
$\coETR$.

Our proof of hardness for $\ExistsESS$ combines ideas of the
$\Pip_2$-completeness proof of the problem \textsc{MinmaxClique} by Ko
and Lin~\cite{Chapter:KoLin1995}, the reduction from the complement of
\textsc{MinmaxClique} to $\ExistsESS$ for two-player games by
Conitzer~\cite{MOR:Conitzer2019}, and the direct translation of
solutions of a polynomial system to strategies of a game by
Hansen~\cite{TCS:Hansen2019}, in addition to new ideas.

A strongly related concept to an ESS is that of a locally superior
strategy (LSS) which is equivalent to an ESS having a uniform invasion
barrier~\cite{JMB:Palm1984}. For the case of two-player games these
concepts coincide~\cite{JTB:HofbauerSS1979}, but they differ for
multi-player games~\cite{RePEc:Milchtaich2008}. Analogously to the
case of ESS we consider the two computational problems $\ExistsLSS$
and $\IsLSS$ and prove the same results for these as for $\ExistsESS$
and $\IsESS$.

We leave the problem of determining the precise computational
complexity of $\ExistsESS$ and $\ExistsLSS$ as an interesting open
problem. The class $\cExistsForallR$ is the natural real complexity
class generalization of $\Sigmap_2$. Together with
$\Sigmap_2$-completeness of $\ExistsESS$ for the setting of two-player
games, this might lead one to expect that $\ExistsESS$ should be
$\cExistsForallR$-hard for multi-player games. However, a basic
property of the set of evolutionary stable strategies is that any ESS
is an isolated point in the space of
strategies~\cite[Proposition~3]{IGTR:AccinelliMO2019}, which means
that the set of evolutionary stable strategies is always a
\emph{discrete} set. Expressing $\ExistsESS$ in
$\theory_{\exists\forall}(\RR)$, the universal quantifier range over
all potential ESS and the existential quantifier over potential
invading strategies. The fact that the set of ESS is a discrete set
could possibly mean that the universal quantifier could be made
discrete as well. We also note that we do not even know whether
$\ExistsESS$ is hard for $\cETR$, which is clearly a prerequisite for
$\cExistsForallR$-hardness.

\section{Preliminaries}
\label{SEC:Preliminaries}

\subsection{Strategic Form Games}
We present here basic definitions concerning strategic form games,
mainly to establish our notation. A finite $m$-player strategic form
game $\calG$ is given by finite sets $S_1,\dots,S_m$ of actions
(\emph{pure strategies}) together with \emph{utility functions}
$u_1,\dots,u_m : S_1 \times \dots \times S_m \rightarrow \RR$. A
choice of an action $a_i\in S_i$ for each player together form a pure
strategy profile $a=(a_1,\dots,a_m)$. Let $\Delta(S_i)$ denote the set
of probability distributions on $S_i$.  A \emph{(mixed) strategy} for
player~$i$ is then an element $x_i \in \Delta(S_i)$. We may
conveniently identify an action $a_i$ with the strategy that assigns
probability~1 to~$a_i$. A strategy $x_i$ for each player $i$ together
form a strategy profile $x=(x_1,\dots,x_m)$. For fixed $i$ we denote
by $x_{-i}$ the \emph{partial} strategy profile
$(x_1,\dots,x_{i-1},x_{i+1},\dots,x_m)$ for all players except
player~$i$, and if $x'_i \in \Delta(S_i)$ we denote by $(x'_i;x_{-i})$
the strategy profile $(x_1,\dots,x_{i-1},x'_i,x_{i+1},\dots,x_m)$.
The utility functions extend to strategy profiles by letting
$u_i(x)=\Exp_{a \sim x}u_i(a_1,\dots,a_m)$. We shall also refer to
$u_i(x)$ as the \emph{payoff} of player~$i$. A strategy profile $x$ is
a Nash equilibrium (NE) if $u_i(x) \geq u_i(x'_i;x_{-i})$ for all~$i$
and all $x'_i \in \Delta(S_i)$. Every finite strategic form game
$\calG$ has an NE~\cite{AM:Nash51}.

In this paper we shall only consider \emph{symmetric} games. The game
$\calG$ is symmetric if all players have the same set $S$ of actions
and where the utility function of a given player depends only on the
action of that player (and not the identity of the player) together
with the \emph{multiset} of actions of the other players. More
precisely we say that $\calG$ is \emph{symmetric} if there is a finite
set $S$ such that $S_i=S$, for every $i \in [m]$, and such that for
every permutation $\pi$ on $[m]$, every $i \in [m]$ and every
$(a_1,\dots,a_m) \in S^m$ it holds that
$u_i(a_1,\dots,a_m)=u_{\pi^{-1}(i)}(a_{\pi(1)},\dots,a_{\pi(m)})$. It
follows that a symmetric game $\calG$ is fully specified by $S$ and
$u_1$; for simplicity we let $u=u_1$. A strategy profile
$x=(x_1,\dots,x_m)$ is symmetric if $x_1=\dots=x_m$. If a symmetric
strategy profile $x$ is an NE it is called a symmetric NE (SNE). Every
finite strategic form symmetric game $\calG$ has a
SNE~\cite{AM:Nash51}.

A single strategy $\sigma \in \Delta(S)$ defines the symmetric
strategy profile $\sigma^m$. More generally, given
$\sigma, \sigma_1,\dots,\sigma_r\in \Delta(S)$ and
$m_1,\dots,m_r\geq 1$ with $m_1+\dots+m_r=m-1$, we denote by
$(\sigma; \sigma_1^{m_1},\dots,\sigma_r^{m_r})$ a strategy profile
where player~1 is playing using strategy $\sigma$ and $m_i$ of the
remaining players are playing using strategy $\sigma_i$, for
$i=1,\dots,r$. By the assumptions of symmetry, the payoff
$u(\sigma;\sigma_1^{m_1},\dots,\sigma_{r}^{m_r})$ is well defined.

\subsection{Evolutionary Stable Strategies}
Our main object of study is the notion of evolutionary stable
strategies as defined by Maynard Smith and
Price~\cite{Nature:MaynardSmithP1973} for 2-player games and
generalized to multi-player games by Palm~\cite{JMB:Palm1984} and
Broom, Cannings, and Vickers~\cite{BMB:BroomCV1997}. We follow below
the definition given by Broom~et~al.
\begin{definition}
  \label{DEF:ESS}
Let $\calG$ be a symmetric game given by $S$ and $u$. Let
$\sigma,\tau \in \Delta(S)$. We say that $\sigma$ is evolutionary
stable (ES) against $\tau$ if there is $\eps_{\tau}>0$ such that for all
$0<\eps < \eps_{\tau}$ we have
\begin{equation}
  u(\sigma ; \tau_\eps^{m-1}) > u(\tau; \tau_\eps^{m-1}) \enspace ,
  \label{EQ:ESdef}
\end{equation}
where $\tau_\eps = \eps \tau + (1-\eps)\sigma$ is the strategy that
plays according to $\tau$ with probability $\eps$ and according to
$\sigma$ with probability $1-\eps$. We say that $\sigma$ is an
evolutionary stable strategy (ESS) if $\sigma$ is ES against every
$\tau \neq \sigma$. If $\sigma$ is not ES against $\tau$ we also say
that $\tau$ \emph{invades} $\sigma$.
\end{definition}

The supremum over $\eps_{\tau}$ for which Equation~(\ref{EQ:ESdef})
holds is called the \emph{invasion barrier} for $\tau$.  If $\sigma$
is an ESS and there exists $\eps_\sigma>0$ such that for all
$\tau\neq\sigma$ the invasion barrier $\eps_\tau$ for $\tau$ satisfies
$\eps_\tau \geq \eps_\sigma$, we say that $\sigma$ is an ESS with
\emph{uniform} invasion barrier $\eps_\sigma$. For 2-player games any
ESS has a uniform invasion
barrier~\cite{JTB:HofbauerSS1979}. Milchtaich~\cite{RePEc:Milchtaich2008}
give a simple example of an ESS in a 4-player game without a uniform
invasion barrier.

The following simple lemma due to Broom~et~al.~\cite{BMB:BroomCV1997}
provides a useful alternative characterization of an ESS.
\begin{lemma}
  \label{LEM:ESSdef}
  A strategy $\sigma$ is ES against $\tau$ if and only if there exists
  $0\leq j <m$ such that
  $u(\sigma; \tau^j, \sigma^{m-1-j}) > u(\tau; \tau^j, \sigma^{m-1-j})$ and that for all $0\leq i<j$,   $u(\sigma; \tau^i, \sigma^{m-1-i}) = u(\tau; \tau^i, \sigma^{m-1-i})$.
\end{lemma}
For the case of 2-player games, this alternative characterization is
actually the original definition of an ESS given by Maynard Smith and
Price~\cite{Nature:MaynardSmithP1973}, and the definition of an ESS we
use was stated for the case of 2-player games by Taylor and
Jonker~\cite{MBS:TaylorJ1978}. A straightforward corollary of the
characterization is that if $\sigma$ is an ESS then $\sigma^m$ is a
SNE.

By the \emph{support} of an ESS $\sigma$, Supp($\sigma$), we refer to
the set of pure strategies $i$ that are played with non-zero
probability under the strategy $\sigma$.

\subsection{Locally Superior Strategies}

A concept strongly related to evolutionary stable strategies is that
locally superior strategies.
\begin{definition}
  \label{DEF:LocallySuperior}
  A strategy $\sigma$ is a locally superior strategy (LSS) if there
  exists $\eps_\sigma>0$ such that
  $u(\sigma;\tau^{m-1}) > u(\tau^{m})$ for all $\tau$ satisfying
  $0<\norm{\sigma-\tau}_1<\eps_\sigma$.
\end{definition}
It is easy to see that $\sigma$ is locally superior if and only if
$\sigma$ is an ESS with a uniform invasion
barrier~\cite{JMB:Palm1984}. Indeed, with $\tau_\eps$ as in
Definition~\ref{DEF:ESS} we have $u(\tau_\eps^m)=\eps u(\tau; \tau_\eps^{m-1}) + (1-\eps)u(\sigma; \tau_\eps^{m-1})$, from which it follows that for $0<\eps<1$ we have
$u(\sigma; \tau_\eps^{m-1}) > u(\tau; \tau_\eps^{m-1})$ if and only if
$u(\sigma; \tau_\eps^{m-1}) > u(\tau_\eps^m)$.

Since as stated, in the case of 2-player games any ESS has a uniform
invation barrier, it follows that the notions of ESS and LSS coincide
for 2-player games.

\subsection{Real Computational Complexity}
While we are mainly interested in the computational complexity of
discrete problems, it is useful to discuss a model of computation
operating on real-valued input. We use this to define the complexity
class $\existsDOP \coETR$, used to formulate our main result.  Alternatively we may
simply define this class in terms of a restriction of the decision
problem for the first-order theory of the reals, as explained in the
next subsection. The reader may thus defer reading this subsection.

A standard model for studying computational complexity in the setting
of reals is that of Blum-Shub-Smale (BSS)
machines~\cite{BAMS:BlumSS1989}. A BSS machine takes a vector $x \in \RR^n$
as an input and performs arithmetic operations and comparisons at
unit cost. In addition the machine may be equipped with a finite set
of real-valued \emph{machine constants}. In this way a BSS machine
accepts a \emph{real language} $L \subseteq \RR^\infty$, where
$\RR^\infty = \bigcup_{n\geq 0} \RR^n$. Imposing polynomial time
bounds we obtain the complexity classes $\PTIME_\RR$ and $\NP_\RR$ for
deterministic and nondeterministic BSS machines, respectively, forming
real-valued analogues of $\PTIME$ and
$\NP$. Cucker~\cite{TCJ:Cucker1993} defined the real analogue
$\PH_\RR$ of the polynomial time hierarchy formed by the classes
$\mathrm{\Sigma}^\RR_k$ and $\mathrm{\Pi}^\RR_k$, for $k\geq 1$. The class
$\mathrm{\Sigma}^\RR_{k+1}$ may be defined as real languages accepted by a
nondeterministic \emph{oracle} BSS machine in polynomial time using an
oracle language from $\mathrm{\Sigma}^\RR_k$ with $\mathrm{\Sigma}^\RR_1=\NP_\RR$, and
$\mathrm{\Pi}^\RR_k$ is simply the class of complements of languages of
$\mathrm{\Sigma}^\RR_k$. For natural problems such as \textsc{TSP} or
\textsc{Knapsack} with real-valued input the search space remains
discrete. Goode~\cite{JSL:Goode1994} introduced the notion of
\emph{digital nondeterminism} (cf.\ \cite{MST:CuckerM96}) restricting
nondeterministic guesses to the set $\zo$, which when imposing
polynomial time bounds define the class $\DNP_\RR$. One may also
define a polynomial hierarchy based on digital nondeterminism giving
rise to classes $\DSigma^\RR_k$ and $\DPi^\RR_k$, for $k\geq 1$.

Another convenient way to define the classes described above is by
means of complexity class \emph{operators} (cf.\
\cite{CCC:Zachos1986,TCS:BorchertS2001}). Here we shall consider
existential or universal quantifiers over either real-valued or
Boolean variables whose number is bounded by a polynomial. For a real
complexity class $\calC$, define $\existsROP\calC$ as the class of real
languages $L$ for which there exists $L' \in \calC$ and a polynomial
$p$ such that $x \in L$ if and only if
$\exists y \in \RR^{\leq p(\Abs{x})} : \pairing{x}{y} \in L'$. For a
real (or discrete) complexity class $\calC$, define $\existsDOP\calC$
as the class of real (or discrete) languages $L$ for which there
exists $L' \in \calC$ and a polynomial $p$ such that $x \in L$ if and
only if $\exists y \in \zo^{\leq p(\Abs{x})} : \pairing{x}{y} \in
L'$. Replacing existential quantifiers with universal quantifiers we
analogously obtain definitions of classes $\forallROP\calC$ and
$\forallDOP\calC$. We now have that
$\mathrm{\Sigma}^\RR_{k+1}=\existsROP\mathrm{\Pi}^\RR_{k}$,
$\DSigma^\RR_{k+1}=\existsDOP\DPi^\RR_{k}$, as well as
$\Sigmap_{k+1}=\existsDOP\Pip_k$, for $k\geq 1$. We shall also
consider mixing real and discrete operators. In such cases one may not
always have an equivalent definition in terms of oracle machines. For
instance, while $\existsROP\coDNP_\RR = \NP_\RR^{\DNP_\RR}$ we can only
prove the inclusion $\existsDOP\coNP_\RR \subseteq \DNP_\RR^{\NP_\RR}$
and in particular we do not know if
$\NP_\RR \subseteq \existsDOP\coNP_\RR$.

To study discrete problems we define the \emph{Boolean part} of a real
language $L \subseteq \RR^\infty$ as $\BP(L) = L \cap \zo^*$ and of
real complexity classes $\calC$ as
$\BP(\calC)=\{\BP(L) \mid L \in \calC\}$. The Boolean part of a real
complexity class is thus a discrete complexity class and may be compared
with other discrete complexity classes defined for instance using
Turing machines. Furthermore, since we are interested in
\emph{uniform} discrete complexity we shall disallow machine
constants. Indeed, a single real number may encode an infinite
sequence of discrete advice strings, which for instance implies
that $\PTIME/\poly \subseteq \BP(\PTIME_\RR)$. For a class $\calC$
defined above we denote by $\calC^0$ the analogously defined class
without machine constants. Several classes given by Boolean parts of
constant-free real complexity are defined specifically in the
literature. Most prominently is the class $\BP(\NP^0_\RR)$ which also
captures the complexity of the existential theory of the reals. It has
been named $\cETR$ by Schaefer and
Štefankovič~\cite{TOCS:SchaeferS2017} as well as $\mathrm{NPR}$ by
Bürgisser and Cucker~\cite{FCM:BurgisserC2009}; we shall use the
former notation $\cETR$. We further let
$\coETR=\BP(\coNP^0_\RR)$ as well as
$\cExistsForallR=\BP(\mathrm{\Sigma}^{\RR,0}_2)=\existsROP\coETR$ and
$\cForallExistsR=\BP(\mathrm{\Pi}^{\RR,0}_2)=\forallROP\cETR$. We shall in
particular be interested in the class $\existsDOP\coETR$. Clearly,
from the definitions above we have that this class contains both the
familiar classes $\coETR$ and $\Sigmap_2$ and is itself contained in
$\cExistsForallR$. In fact $\existsDOP\coETR$ contains the class
$(\Sigmap_2)^\PosSLP$, where $\PosSLP$ is the problem of deciding
whether an integer given by a division free arithmetic circuit is
positive, as introduced by
Allender~et~al.~\cite{SICOMP:AllenderBKM2009}. This follows since
$\PTIME^\PosSLP=\BP(\PTIME^0_\RR)$~\cite[Proposition~1.1]{SICOMP:AllenderBKM2009},
and thus
\[
  \begin{split}
  (\Sigmap_2)^\PosSLP &= \existsDOP \forallDOP \PTIME^\PosSLP =
  \existsDOP\forallDOP \BP(\PTIME^0_\RR) \\&\subseteq \existsDOP \BP( \forallROP \PTIME^0_\RR) = \existsDOP \BP(\coNP^0_\RR) = \existsDOP \coETR \enspace .
\end{split}
\]

\subsection{The First-Order Theory of the Reals}
The discrete complexity classes $\BP(\Sigma^{\RR,0}_k)$ and
$\BP(\Pi^{\RR,0}_k)$ may alternatively be characterized using the
decision problem for the first-order theory of the reals. We denote by
$\theory(\RR)$ the set of all true first-order sentences over the
reals. We shall consider the restriction to sentences in \emph{prenex normal form}
\begin{equation}
  \left(Q_1 x_1 \in \RR^{n_1}\right)\cdots\left(Q_k x_k \in \RR^{n_k}\right) \varphi(x_1,\dots,x_k) \enspace ,
\end{equation}
where $\varphi$ is a quantifier free Boolean formula of equalities and
inequalities of polynomials with integer coefficients, where each
$Q_i$ is one of the quantifiers $\exists$ or $\forall$, typically
alternating, and gives rise to~$k$ blocks of quantified variables. The
restriction of $\theory(\RR)$ to formulas in prenex normal form with
$k$ being a fixed constant and also $Q_1=\exists$ is complete for
$\BP(\Sigma^{\RR,0}_k)$; when instead $Q_1=\forall$ it is complete for
$\BP(\Pi^{\RR,0}_k)$. In particular, the \emph{existential theory of
  the reals} $\Etheory(\RR)$, where $k=1$ and $Q_1=\exists$, is complete
for $\cETR$. Similarly $\theory_{\forall\exists}(\RR)$ where $k=2$ and
$Q_1=\forall$ is complete for $\cForallExistsR$; when we furthermore
restrict the first quantifier block to Boolean variables the problem
becomes complete for $\existsDOP \coETR$.

\subsection{Real Polynomials with Discrete Quantification}
In this section we shall prove that the following problem,
$\forallD\HOMFOURFEAS(\Simplex)$, is complete for the complexity class
$\forallDOP \cETR$. In Section~\ref{SEC:ESS} and Section~\ref{SEC:LSS}
we use the complement of this problem to prove our main results of
$\existsDOP \coETR$-hardness of $\ExistsESS$ and of $\ExistsLSS$.

Denote by $\Simplex^n \subseteq \RR^{n+1}$ the $n$-simplex
$\{x \in \RR^{n+1} \mid x\geq 0 \wedge \sum_{i=1}^{n+1} x_i = 1\}$ and
similarly by $\CornerSimplex^n \subseteq \RR^n$ the corner
$n$-simplex
$\{x \in \RR^n \mid x\geq 0 \wedge \sum_{i=1}^n x_i \leq 1\}$.
\begin{definition}[$\forallD\HOMFOURFEAS(\Simplex)$]
  \label{DEF:forallDHOMFOURFEASSimplex}
  For the problem $\forallD\HOMFOURFEAS(\Simplex)$ we are given as
  input rational coefficients $a_{i,\alpha}$, where
  $i \in \{0,\dots,n\}$ and $\alpha \in [m]^4$, forming the polynomial
\[
  F(y,z) = F_0(z) + \sum_{i=1}^n y_i F_i(z) \enspace ,
\]
where
\[
  F_i(z) = \sum_{\alpha \in [m]^4} a_{i,\alpha} \prod_{j=1}^4 z_{\alpha_j} \text{ , for } i=0,\dots,n \enspace .
\]
We are to decide whether for all $y \in \zo^n$
there exists $z \in \Simplex^{m-1}$ such that $F(y,z)=0$.
\end{definition}

Analogously to the fact that a matrix representing a quadratic form
may be assumed to be symmetric, we may assume that the polynomials of
Definition~\ref{DEF:forallDHOMFOURFEASSimplex} are \emph{symmetrized}.
\begin{definition}
  \label{DEF:Symmetrization}
  For $\alpha \in [m]^4$ and a permutation $\pi$ on $[4]$, define
  $\pi\cdot\alpha \in [m]^4$ by $(\pi\cdot\alpha)_i=\alpha_{\pi(i)}$.
  We say that a homogenous polynomial $G$ given in the form
\begin{equation}
\label{EQ:SymmetrizationPolynomial}
  G(z) = \sum_{\alpha \in [m]^4} b_{\alpha} \prod_{j=1}^4 z_{\alpha_j}
\end{equation}
is \emph{symmetrized} if $b_{\alpha} = b_{\pi\cdot\alpha}$ for all $\alpha$ and $\pi$.
\end{definition}
\begin{lemma}
\label{LEM:Symmetrization}
Any homogeneous polynomial $G$ in the form of
Equation~(\ref{EQ:SymmetrizationPolynomial}) is, as a function, equal
to a symmetrized homogeneous polynomial $H$ in the same form.
\end{lemma}
\begin{proof}
  Let $c_{\alpha} = \frac{1}{24}\sum_{\pi} b_{\pi \cdot \alpha}$ for
  all $\alpha$, and define $H$ by
  \[
    H(z) = \sum_{\alpha \in [m]^4} c_{\alpha} \prod_{j=1}^4 z_{\alpha_j}  \enspace .
  \]
  Then $H$ is clearly symmetrized and it holds that $G(z)=H(z)$ for
  all $z$.
\end{proof}

We next turn to the proof of $\forallDOP \cETR$-hardness of
$\forallD\HOMFOURFEAS(\Simplex)$. The proof is mainly a combination of
existing ideas and proofs, and the reader may thus defer reading it.

\begin{theorem}
  \label{THM:forallHOM4FEAS}
  The problem $\forallD\HOMFOURFEAS(\Simplex)$ is complete for
  $\forallDOP\cETR$, and remains $\forallDOP\cETR$-hard even with the promise that
  for all $y \in \zo^n$ and $z \in \RR^m$ it holds that
  $F(y,z)\geq 0$.
\end{theorem}

\begin{proof}
  We shall prove hardness of $\forallD\HOMFOURFEAS(\Simplex)$ by
  describing a general reduction from a language $L$ in
  $\forallDOP \cETR$ in several steps making use of reductions that
  proves several problems involving real polynomials
  $\cETR$-hard. Consider first the standard complete problem $\QUAD$
  for $\cETR$ which is that of deciding if a system of multivariate
  quadratic polynomials have a common
  root~\cite{Book:BlumCSS1998,TOCS:SchaeferS2017}. The general
  reduction from a language $L$ in $\cETR$ to $\QUAD$ works by
  treating the input $x$ as variables and computes, based \emph{only}
  on $\Abs{x}$ and not the actual value of $x$, a system of quadratic
  polynomials $q_i(x,y)$, $i=1,\dots,\ell$, where
  $y \in \RR^{p(\Abs{x})}$ for some polynomial $p$. The system has the
  property that for all $x$ it holds that $x \in L$ if and only if
  there exists $y$ such that $q_i(x,y)=0$, for all~$i$.

  Suppose now that $L \in \forallDOP \cETR$. Then there is $L'$ in
  $\cETR$ and a polynomial $p$ such that $x \in L$ if and only if
  $\forall y \in \zo^{p(\Abs{x})} : \pairing{x}{y} \in L'$. On input
  $x$ we may apply the reduction from $L'$ to $\QUAD$ and in this way
  obtain a system of quadratic equations $q_i(x,y,z)$,
  $i=1,\dots,\ell_1$ where $z \in \RR^{p_1(\Abs{x})}$, for some
  polynomial $p_1$, such that $\pairing{x}{y} \in L'$ if and only if
  there exists $z \in \RR^{p_1(\Abs{x})}$ such that $q_i(x,y,z)=0$ for
  all~$i$. At this point we may just treat $x$ as fixed constants, and
  we view the system as polynomials in variables $(y,z)$, suppressing
  the dependence on $x$ in the notation. Define $n=p(\Abs{x})$. We
  next introduce additional existentially quantified variables
  $w \in \RR^n$, substitute $w_i$ for $y_i$ in all polynomials, and
  then add new polynomials $w_i-y_i$, for $i \in [n]$. Renaming
  polynomials and bundling the existentially quantified variables we
  now have a system of quadratic polynomials $q_i(y,z)$,
  $i\in[\ell_2]$ where $z \in \RR^{m_2}$ and $m_2 \leq p_2(\Abs{x})$
  for some polynomial $p_2$, such that $x \in L$ if and only if
\[
  \forall y \in \zo^n \exists z \in \RR^{m_2} \forall i \in [\ell_2] : q_i(y,z)=0 \enspace ,
\]
and where each polynomial $q_i$ depends on at most~1 coordinate of
$y$.

For the next step we use that $\QUAD$ remains $\cETR$-hard when asking
for a solution in the unit ball~\cite{GD:Schaefer2009}, or analogously
in the corner simplex~\cite{TCS:Hansen2019}. Applying the reduction of
\cite[Proposition~2]{TCS:Hansen2019} we first rewrite each variable
$z_i$ as a difference $z_i=z^+_i-z^-_i$ of two non-negative real
variables $z^+_i$ and $z^-_i$ and then introduce additional
existentially quantified variables $w_0,\dots,w_t$ for suitable
$t=O(\log \tau+m_2)$, where $\tau$ is the maximum bitlength of the
coefficients of the given system. Then polynomials are added that
together implement $t$ steps of repeated squaring of $\frac{1}{2}$,
i.e.\ we add polynomials $w_t-\frac{1}{2}$, and $w_{j-1}-w_j^2$, for
$j \in [t]$, which means that any solution must then have
$w_0=2^{-2^t}$.

In the polynomial system we now substitute each occurrence of $z_i$ by
$(z^+_i-z^-_i)/w_0$ and afterwards multiply by $w_0^2$ in every
polynomial where a subtitution occurred, in order to clear $w_0$ from
the denominators. For suitable $t$ this implies that if for fixed $y$,
the given system of polynomials has a solution $z \in \RR^{m_2}$, then
the transformed system has a solution $(z^+,z^-,w)$ in
$\CornerSimplex^{2m_2+t+1}$.  Note that, since the variables $y_i$ are
not divided by $w_0$, the polynomials are no longer of degree at
most~2 after multiplication with $w_0^2$.  However, they remain of
degree at most~$2$ in the variables $(z^+,z^-,w)$.

Again, renaming polynomials and bundling the existentially quantified
variables we now have a system of polynomials $q_i(y,z)$,
$i\in[\ell_3]$, where $z \in \RR^{m_3}$ and $m_3 \leq p_3(\Abs{x})$
for some polynomial $p_3$, such that $x \in L$ if and only if
\[
  \forall y \in \zo^n \exists z \in \CornerSimplex^{m_3} \forall i \in [\ell_3] : q_i(y,z)=0 \enspace ,
\]
and where each polynomial $q_i$ depends on at most~1 coordinate of
$y$ and is of degree at most~$2$ in the variables~$z$.

The next step simply consists of homogenizing the polynomials in the
existentially quantified variables $z$. For this we simply introduce a
\emph{slack variable} $z_{m_3+1}=1-\sum_{i=1}^{m_3}z_i$ and homogenize
by multiplying terms by $\sum_{i=1}^{m_3+1}z_i$ or
$\sum_{i=1}^{m_3+1}\sum_{j=1}^{m_3+1}z_iz_j$ as needed. Letting $q'_i$
be the homogenization of $q_i$ we now have that $x \in L$ if and only if
\[
  \forall y \in \zo^n \exists z \in \Simplex^{m_3} \forall i \in [\ell_3] : q'_i(y,z)=0 \enspace ,
\]
and where each polynomial $q'_i$ depends on at most~1 coordinate of
$y$ and are homogeneous of degree~2 in the variables~$z$.

For the final step we reuse the idea of the reduction from $\QUAD$ to
$\FOURFEAS$, which merely takes the sum of the squares of every given
polynomial.  Thus we let
\[
  F(y,z)=\sum_{i=1}^{\ell_3} (q'(y,z))^2 \enspace .
\]
We note that $(q'(y,z))^2\geq 0$ for all $y$ and $z$ and is homogeneous
of degree~4 in the variables~$z$. Further, since $y_j^2=y_j$ for any
$y_j\in\zo$ we may replace all occurrences of $y_j^2$ by $y_j$ thereby
obtaining an equivalent polynomial (when $y\in\zo^n$) of the form of
Definition~\ref{DEF:forallDHOMFOURFEASSimplex}. We have that for every
fixed $y\in\zo^n$ and all $z\in\RR^m$ that $F(y,z)=0$ if and only if
$q_i(y,z)=0$ for all $i$. Thus $x \in L$ if and only if
\[
  \forall y \in \zo^n \exists z \in \Simplex^{m_3} F(y,z)=0 \enspace ,
\]
which completes the proof of hardness. Let us also note that the
definition of $F$ guarantees that $F(y,z)\geq 0$ for all $y \in \zo^n$
and $z \in \RR^m$. Since on the other
hand clearly $\forallD\HOMFOURFEAS(\Simplex) \in \forallDOP \cETR$ the result follows.
\end{proof}

As a special case, (when there are no universally quantified
variables) the proof gives a reduction from the $\cETR$-complete
problem $\QUAD$ to the problem $\HOMFOURFEAS(\Simplex)$, where we are
given as input a homogeneous degree~4 polynomial $F(z)$ in $m$
variables with rational coefficients and are to decide whether there
exists $z \in \Simplex^{m-1}$ such that $F(z)=0$. Also, we clearly
have that $\HOMFOURFEAS(\Simplex)$ is a member of $\cETR$ and
therefore have the following result.
\begin{theorem}
  \label{THM:HOM4FEAS}
  The problem $\HOMFOURFEAS(\Simplex)$ is complete for $\cETR$, and
  remains $\cETR$-hard even when assuming that for all $z \in \RR^m$ it
  holds that $F(z)\geq 0$.
\end{theorem}

\section{Complexity of ESS}
\label{SEC:ESS}
In this section we shall prove our results for deciding existence of
an ESS. In the proof we will re-use a trick used by
Conitzer~\cite{MOR:Conitzer2019} for the case of 2-player games, where
by duplicating a subset of the actions of a game we ensure that no ESS
can be supported by any of the duplicated actions, as shown in the
following lemma. Here, by duplicating an action we mean that the utilities
assigned to any pure strategy profile involving the duplicated
action is defined to be equal to the utility for the pure strategy
profile obtained by replacing occurrences of the duplicated action
by the original action. The precise property is as follows.
\begin{lemma}
\label{LEM:NoDuplicate}
Let $\calG$ be an $m$-player symmetric game given by $S$ and $u$. Suppose
that $s,s' \in S$ are such that for all strategies $\tau$ we have
$u(s;\tau^{m-1})=u(s';\tau^{m-1})$.  Then $s$ can not be in the
support of an ESS $\sigma$.
\end{lemma}
\begin{proof}
  Suppose $\sigma$ is a strategy with $s \in \support(\sigma)$. Let
  $\sigma'$ be obtained from $\sigma$ by moving the probability mass
  of $s$ to $s'$. From our assumption we then have
  $u(\sigma;\tau^{m-1})=u(\sigma';\tau^{m-1})$ for all $\tau$. In
  particular we have
  $u(\sigma;\sigma_\eps^{m-1})=u(\sigma';\sigma_\eps^{m-1})$, for all
  $\eps>0$, where $\sigma_\eps$ is given by
  $\sigma_\eps=\eps\sigma'+(1-\eps)\sigma$. This means that $\sigma'$
  invades $\sigma$ and $\sigma$ is therefore not an ESS.
\end{proof}


We now state and prove our first main result.
\begin{theorem}
  \label{THM:ExistsESSHardness}
  $\ExistsESS$ is $\existsDOP\coETR$-hard for 5-player games.
\end{theorem}
\begin{proof}
  We prove our result by giving a reduction from the \emph{complement}
  of the problem $\forallD\HOMFOURFEAS(\Simplex)$ to $\ExistsESS$. It
  follows from Theorem~\ref{THM:forallHOM4FEAS} that the former
  problem is complete for $\existsDOP\coETR$. Thus let $a_{i,\alpha}$
  be given rational coefficients, with $i=0,\dots,n$ and
  $\alpha\in [m]^4$, forming the polynomials $F(y,z)$ and $F_i(z)$,
  for $i=0,\dots,n$ as in
  Definition~\ref{DEF:forallDHOMFOURFEASSimplex}.  We may assume that
  for all $y\in\zo^n$ and all $z\in \RR^m$ it holds that
  $F(y,z)\geq 0$. We may also without loss of generality assume that
  each $F_i$ is \emph{symmetrized} by
  Lemma~\ref{LEM:Symmetrization}. This will ensure that the game
  defined below is well defeined and symmetric.

  We next define a 5-player game $\calG$ based on $F$. The strategy
  set is naturally divided in three parts $S = S_1 \cup S_2 \cup
  S_3$. These are defined as follows.
\begin{equation}
  \begin{aligned}
  S_1 & = \{(i,\alpha,b) \mid i \in \{0,\dots,n\}, ~\alpha \in [m]^4,~b
  \in \zo\}\\
  S_2 & = \{\gamma\}\\
  S_3 & = \{1,\dots,m\}
\end{aligned}
\end{equation}
An action $(i,\alpha,b)$ of $S_1$ thus identifies a term of $F_i$
together with $b \in \zo$, which is supposed to be equal to
$y_i$. When convenient we may describe the actions of $S_1$ by pairs
$(t,b)$, where $t=(i,\alpha)$ for some $i$ and $\alpha$.  The single
action $\gamma$ is used for rewarding inconsistencies in the choices
of $b$ among strategies of $S_1$. Finally, a probability
distribution on $S_3$ will define an input $z$. Let $M=(n+1)m^4$ be
the total number of terms of $F$. Thus $\Abs{S_1}=2M$.

We shall \emph{duplicate} all actions of $S_2 \cup S_3$ and let
duplicates behave exactly the same regarding the utility function
defined below. By Lemma~\ref{LEM:NoDuplicate} it then follows that any
ESS $\sigma$ of $\calG$ must have $\support(\sigma) \subseteq
S_1$. For simplicity we describe the utilities of $\calG$ without the
duplicated actions.

When all players are playing an action of $S_1$ we define
\begin{equation}
  u((t_1,b_1),\dots,(t_5,b_5)) =
  \begin{cases} 2 & \text{if } t_1 \notin \{t_2,\dots,t_5\}\\
    1 & \text{if } t_1 \in \{t_2,\dots,t_5\} \text{ and } t_1 = t_j \Rightarrow b_1 = b_j\\
    0 & \text{otherwise}
  \end{cases} \enspace.
\end{equation}

Before defining the remaining utilities, we consider the payoff of
strategies that play uniformly on the set of terms and according to a
fixed assignment $y$. Define the number $T$ by
\begin{equation}
  T = 1 + \left(1-\frac{1}{M}\right)^4 \enspace .
\end{equation}

\begin{lemma}
  \label{LEM:Sigma-y}
  Let $y \in \zo^n$, let $y_0\in \zo$ be arbitrary, and define $\sigma_y$ to be the
  strategy that plays $(i,\alpha,y_i)$ with
  probability~$\frac{1}{M}$ for all $\alpha$, and the remaining
  strategies with probability~0. Then $u(\sigma_y^5)=T$.
\end{lemma}
\begin{proof}
  Note that $u(\sigma_y^5)-1$ is precisely the probability of the
  intersection of the events $t_1\neq t_j$, where $j=2,\dots,5$, and $t_j$
  is the term chosen by player~$j$. For fixed $t_1$, these events are
  independent and each occurs with probability~$1-\frac{1}{M}$. We thus have
  \[
    u(\sigma_y^5) = 1 + \Pr\left[\bigwedge_{j=2}^4 t_1 \neq t_j\right] = 1 +  \left(1-\frac{1}{M}\right)^4 = T \enspace .
  \]

\end{proof}

We will construct the game $\calG$ in such a way that any ESS $\sigma$
will have $u(\sigma^5)=T$. Making use of
Lemma~\ref{LEM:Sigma-y}, we now define utilities when at
least one player is playing the action~$\gamma$. In case at least two
players are playing~$\gamma$, these players receive utility~$0$ while
the remaining players receive utility~$T$. In case exactly one player
is playing~$\gamma$, the player receives utility~$T+1$ in case there
are two players that play actions $(i,\alpha,b)$ and $(i,\alpha',b')$
with $b\neq b'$; otherwise the player receives utility~$T$. In either
case, when exactly one player is playing~$\gamma$, the remaining
players receive utility~$T$.

We finally define utilities when one player is playing an action from
$S_1$ and the remaining four players are playing an action from
$S_3$. Suppose for simplicity of notation that player $j$ is playing
action $\beta_j \in S_3$, for $j=1,\dots,4$, while player~5 is playing
action $(i,\alpha,b)$. We let player~5 receive utility~$T$. Suppose
that there exists a permutation $\pi$ on~$[4]$ such that
$\beta = \pi \cdot \alpha$. Then, let
$K_\alpha = \abs{\{\pi \cdot \alpha : \pi \in \Sym([4])\}}$ (i.e.\ the
size of the orbit $\Sym([4]) \cdot \alpha$ of $\alpha$ with respect to
the defined group action for the group of permutations on~$[4]$). If
either (i)~$i=0$, or (ii)~ $i>0$, and $b=1$, the first four players
receive utility~$T-\frac{M}{K_\alpha} a_{i,\alpha}$; otherwise they receive
utility~$T$. The first four players also receive utility~$T$ in case
$\beta \neq \pi \cdot \alpha$ for all $\pi$.

The above definition is well defined, since we assumed that each $F_i$
is symmetrized. We observe the following relationship between $\calG$
and $F$.
\begin{lemma}
  \label{LEM:Relationship-G-F}
  Let $y\in \{0,1\}^n$, let $\sigma_y$ be defined as in
    Lemma~\ref{LEM:Sigma-y}, and let
    $z \in \Simplex^{m-1} =\Delta(S_3)$. Then
    $u(z;z^3,\sigma_y)=T-F(y,z)$.
\end{lemma}
\begin{proof}
  Using the above definitions we have
  \[
    \begin{split}
      T-u(z;z^3,\sigma_y) & = \sum_{i=0}^{n}\sum_{\alpha \in [m]^4} \sum_{\beta \in [m]^4} \left(T-u(\beta_1,\beta_2,\beta_3,\beta_4,(i,\alpha,y_i))\right) \frac{1}{M} \prod_{j=1}^4 z_{\beta_j} \\
      & = \sum_{\alpha \in [m]^4} \sum_{\beta \in \Sym([4]) \cdot \alpha} \frac{M}{K_\alpha}\left(a_{0,\alpha} + \sum_{i=1}^n y_i a_{i,\alpha}\right) \frac{1}{M} \prod_{j=1}^4 z_{\beta_j} \\
      & =  \sum_{\alpha \in [m]^4} \left(a_{0,\alpha} + \sum_{i=1}^n y_i a_{i,\alpha}\right)  \prod_{j=1}^4 z_{\alpha_j} \\
      & = F(y,z)
    \end{split}
  \]
\end{proof}

At this point we have only partially specified the utilities of the
game $\calG$; we simply let all remaining unspecified utilities
equal~$T$, thereby completing the definition of $\calG$.

We are now ready to prove that $\calG$ has an ESS if and only if there
exists $y\in\zo^n$ such that $F(y,z)>0$ for all
$z \in \Simplex^{m-1}$. Suppose first that $y \in \zo^n$ exists such
that $F(y,z)>0$ for all $z \in \Simplex^{m-1}$. We define
$\sigma=\sigma_y$ as in Lemma~\ref{LEM:Sigma-y} and show that any
$\tau\neq \sigma$ satisfies the conditions of Lemma~\ref{LEM:ESSdef}
thereby proving that $\sigma$ is an ESS of $\calG$. Suppose that
$\tau \neq \sigma$ invades $\sigma$. Consider first playing $\tau$
against $\sigma^4$. From the proof of Lemma~\ref{LEM:Sigma-y} it
follows that playing a strategy of the form $(i,\alpha,b)$ against
$\sigma^4$ gives payoff~$T$ if $b=y_i$ and otherwise payoff strictly
below~$T$. The strategies of $S_2 \cup S_3$ all give payoff~$T$
against $\sigma^4$. It follows that to invade $\sigma$, $\tau$ can
only play strategies from $S_1$ contained in $\support(\sigma)$. Let
us write $\tau = \delta_1\tau_1+\delta_2\tau_2+\delta_3\tau_3$ as a
convex combination of strategies $\tau_j$ with
$\support(\tau_j)\subseteq S_j$, for $j=1,2,3$. We shall consider
playing $\tau$ against $(\tau,\sigma^3)$ and argue that
$\tau_1=\sigma$ if $\delta_1>0$ and that $\delta_2=0$. Note first that
if a strategy of $S_3$ is played, all players receive utility~$T$, so
we may focus on the case when all players play using strategies from
$S_1 \cup S_2$. Suppose that $\delta_1>0$ and for a term
$t=(i,\alpha)$ of $F$ let $p_{t} = \Pr_{\tau_1}[(t,y_i)]$. We now have
\[
    u(\tau_1;\tau_1,\sigma^3) = 1 + \sum_t p_t(1-p_t)\left(1-\frac{1}{M}\right)^3
    = 1 + \left(1-\frac{1}{M}\right)^3 \sum_t p_t(1-p_t) \enspace .
\]
and by Chebyshev's sum inequality it follows that
\[
\sum_t p_t(1-p_t) \leq
\frac{1}{M}\left(\sum_t{p_t}\right)\left(\sum_t{(1-p_t)}\right)
= 1-\frac{1}{M} \enspace ,
\]
and that equality  holds if and only if
$p_t=\frac{1}{M}$ for all~$t$. Observe also that
\[
  u(\sigma;\tau_1,\sigma^3)= 1 + \sum_t \frac{1}{M}(1-p_t)\left(1-\frac{1}{M}\right)^3 = 1 + \frac{1}{M}(M-1)\left(1-\frac{1}{M}\right)^3 = T\enspace .
\]
Thus if $\tau_1\neq \sigma$, it follows that
$u(\tau_1;\tau_1,\sigma^3)<u(\sigma;\tau_1,\sigma^3)=T$.  Now, since
$\support(\tau_1)\subseteq \support(\sigma)$ when $\delta_1>0$,
playing $\gamma$ can give utility at most~$T$, but also gives
utility~$0$ in case another player plays~$\gamma$ as well.

Combining these observations it follows that unless $\delta_2=0$ and
that $\tau_1=\sigma$ when $\delta_1>0$ we have
$u(\sigma;\tau,\sigma^3)>u(\tau;\tau,\sigma^3)$.  Thus we may now
assume that this is the case, i.e., that
$\tau=\delta_1 \sigma + \delta_3 \tau_3$. From the definition of
$\calG$ we now have that
$u(\tau;\tau^j,\sigma^{4-j}) \leq u(\sigma;\tau^j,\sigma^{4-j})= T$,
for $j=1,2,3$. For $\tau$ to invade $\sigma$ it is thus required that
$u(\tau; \tau^4) \geq u(\sigma; \tau^4)$, and it follows from the
definition of $\calG$ that this is equivalent to
$u(\tau_3; \tau_3^3,\sigma)\geq T$. Now
$\tau_3 \in \Delta(S_3) = \Simplex^{m-1}$ and by assumption we have
$F(y,\tau_3)>0$. Furthermore we have
$u(\tau_3; \tau_3^3,\sigma) = T - F(y,\tau_3)$ and thus
$u(\tau_3; \tau_3^3,\sigma) < T$, which means $\sigma$ is actually ES
against $\tau$.

Suppose now on the other hand that $\sigma$ is an ESS of
$\calG$. First, since we duplicated the actions of $S_2 \cup S_3$, it
follows from Lemma~\ref{LEM:NoDuplicate} that
$\support(\sigma)\subseteq S_1$. We next show that for all terms $t$,
if $\Pr_{\sigma}[(t,b)]>0$, then unless $Pr_{\sigma}[(t,1-b)]=0$, $\sigma$ can be
invaded. Suppose that $t$ is a term of $F$, let $p_0=\Pr_{\sigma}[(t,0)]$ and
$p_1=\Pr_{\sigma}[(t,1)]$, and suppose that $p_0>0$ and $p_1>0$. Suppose
without loss of generality that $p_0\geq p_1$. Note now that
\[
  u((t,0);\sigma^4) - u((t,1);\sigma^4) = (1-p_1)^4-(1-p_0)^4 \geq 0 \enspace ,
\]
which can be seen by noting that that the left hand side of the
equality does not change when replacing all utilities of~2 by~1. Similarly
\[
  u((t,0);(t,0),\sigma^3) - u((t,1);(t,1),\sigma^3) =
  (1-p_1)^3-(1-p_0)^3 \geq 0 \enspace .
\]
Define the strategy $\sigma'$ from $\sigma$ by playing the
strategy~$(t,0)$ with probability~$p=p_0+p_1$, the strategy $(t,1)$
with probability~$0$, and otherwise according to $\sigma$. Then
\[
  \begin{split}
  u(\sigma';\sigma^4)-u(\sigma^5) &= (p_0+p_1)u((t,0);\sigma^4)-p_0u((t,0);\sigma^4)-p_1u((t,1);\sigma^4) \\ &= p_1(u((t,0);\sigma^4) - u((t,0);\sigma^4)) \geq 0 \enspace .
\end{split}
\]
By definition, $u((t,0);(t,1), \sigma_3) = u((t,1);(t,0), \sigma_3) = 0$, and we thus have
\[
\begin{split}
u&(\sigma';\sigma', \sigma^3) - u(\sigma;\sigma', \sigma^3) \\ &= (p_0+p_1)^2u((t,0);(t,0), \sigma^3) - p_0(p_0+p_1)u((t,0);(t,0), \sigma^3)\\
&=(p_0p_1+p_1^2)u((t,0);(t,1), \sigma^3) > 0 .
\end{split}
\]
which means that $\sigma'$ invades $\sigma$. Since $\sigma$ is an ESS,
this means that for each term $t$ there is $b_t\in \zo$ such that
$\sigma$ plays $(t,1-b_t)$ with probability~$0$. Let
$p_t=\Pr_{\sigma}[(t,b_t)]$ for all $t$. Defining the function
$h : \RR \rightarrow \RR$ by $h(p) = (1-p)^4$ we now have
\[
  u(\sigma^5) = 1+\sum_t p_t h(p_t) \enspace .
\]
Suppose there exists terms $t$ and $t'$ such that $p_t <
p_{t'}$. Since $h$ is strictly decreasing on $[0,1]$ we then then have
$h(p_t) > h(p_{t'})$, and therefore
$p_t h(p_t) + p_{t'} h(p_{t'}) < p_{t'} h(p_t) + p_t h(p_{t'})$.
Define $\sigma'$ to play $t$ with probability $p_{t'}$, $t'$ with
probability $p_t$, and otherwise according to $\sigma$. We then have
\[
  u(\sigma^5)) - u(\sigma';\sigma^4)) =  p_t(h(p_t)-h(p_{t'}))+p_{t'}(h(p_{t'})-h(p_t)) < 0 \enspace ,
\]
which means that $\sigma'$ invades $\sigma$. Since $\sigma$ is an ESS
this means that $p_t=\frac{1}{M}$ for all $t$. From the proof of
Lemma~\ref{LEM:Sigma-y} it then follows that
$u(\sigma^5) = T$.

Suppose now that there exists $i \in [n]$ and $\alpha,\alpha'$ such that
$b_{(i,\alpha)} \neq b_{(i,\alpha')}$. But then
$u(\gamma;\sigma^4)>T = u(\sigma^5)$, which means that $\gamma$
invades $\sigma$. Since $\sigma$ is an ESS there must exist $y\in\zo^n$ (and some $y_0\in\zo$) such that
$\sigma=\sigma_y$, using the notation of
Lemma~\ref{LEM:Sigma-y}.

Finally, let $z \in \Simplex^{m-1}=\Delta(S_3)$. By definition of $u$ we have
$u(z;z^j,\sigma^{4-j})=T=u(\sigma;z^j,\sigma^{4-j})$, for all
$j\in\{0,1,2\}$. Next
$u(z;z^3,\sigma)=T-F(y,z)$ while we have $u(\sigma;z^3,\sigma)=T$. For
$\sigma$ to be ES against $z$ we must thus have $F(y,z)>0$, and this
concludes the proof.
\end{proof}

The best upper bound on the complexity of $\ExistsESS$ we know is
membership of $\cExistsForallR$ which easily follows from either
Definition~\ref{DEF:ESS} or from the characterization of
Lemma~\ref{LEM:ESSdef}. For the simpler problem $\IsESS$ of
determining whether a given strategy is an ESS we can fully
characterize its complexity.
\begin{theorem}
\label{THM:IsESSCompleteness}
  $\IsESS$ is $\coETR$-complete for 5-player games.
\end{theorem}
\begin{proof}
  Clearly $\IsESS$ belongs to $\coETR$ by the characterization of
  Lemma~\ref{LEM:ESSdef}. To show $\coETR$-hardness we reduce from the
  \emph{complement} of the problem $\HOMFOURFEAS(\Simplex)$ to
  $\IsESS$.  It follows from Theorem~\ref{THM:HOM4FEAS} that the
  former problem is complete for $\coETR$. From $F$ we construct the
  game $\calG$ as in the proof of Theorem~\ref{THM:ExistsESSHardness}
  letting $n=0$. We let $\sigma$ be the uniform distribution on the
  set of actions $(0,\alpha,0)$, where $\alpha \in [m]^4$. It then
  follows from the proof of Theorem~\ref{THM:ExistsESSHardness} that
  $\sigma$ is an ESS of $\calG$ if and only if $F(z)>0$ for all
  $z \in \Simplex^{m-1}$. Since we may assume that $F(z)\geq 0$ for
  all $z \in \RR^m$ this completes the proof.
\end{proof}

\section{Complexity of LSS}
\label{SEC:LSS}

In this section we extend our results for deciding existence of an ESS
to that of deciding existence of a LSS. The results are obtained by
reusing the reduction from the complement of
$\forallD\HOMFOURFEAS(\Simplex)$ to $\ExistsESS$ as a reduction to
$\ExistsLSS$. Since any LSS is also an ESS it will suffice to prove
that if the game constructed has an ESS it also has a LSS. While the
proofs of this section thus subsumes parts of the proof of
Theorem~\ref{THM:ExistsESSHardness} (and
Theorem~\ref{THM:IsESSCompleteness}), we presented those separately,
since the proofs the proofs of this section are more involved.

We state our second main result below.
\begin{theorem}
  \label{THM:ExistsLSSHardness}
  $\ExistsLSS$ is $\existsDOP\coETR$-hard for 5-player games.
\end{theorem}
Like for the problem $\ExistsESS$ the best upper bound on the
complexity of $\ExistsLSS$ we know is membership of $\cExistsForallR$
which follows directly from
Definition~\ref{DEF:LocallySuperior}.

The proof of Theorem~\ref{THM:ExistsLSSHardness} gives as a special
case a reduction from the complement of $\HOMFOURFEAS(\Simplex)$ to
the problem $\IsLSS$, analogously to the proof of
Theorem~\ref{THM:IsESSCompleteness}, thereby showing that $\IsLSS$ is
$\coETR$-hard. On the other hand, proving $\coETR$-membership of
$\IsLSS$ is not as simple as for $\IsESS$. This is because
Definition~\ref{DEF:LocallySuperior} defining that $\sigma$ is LSS
involves a leading existential quatifier in front of the universal
quantification over other strategies $\tau$, and we have no
alternative definition without such a leading existential quantifier,
unlike the case of ESS where this was given by
Lemma~\ref{LEM:ESSdef}. The existential quantifier is however just
used for expressing universal quatification over sufficiently close
strategies $\tau$ to $\sigma$, and it follow by a general result of
B{\"u}rgisser and Cucker~\cite[Theorem~9.2]{FCM:BurgisserC2009} that this
can be done in $\coETR$.
\begin{lemma}[B{\"u}rgisser and Cucker]
  \label{LEM:EliminateForallSmall}
  The problem of deciding if a sentence of the form
  \[
    \exists \epsilon_0 > 0\ \forall \eps \in (0,\eps_0)\ \forall x :
    \varphi(\eps,x)
  \]
  is true, for a given a quantifier-free formula over the reals
  $\varphi(\eps,x)$ belongs to~$\coETR$.
\end{lemma}
With this in hand, membership of $\IsLSS$ in $\coETR$ is straightforward and
combined with the result of Theorem~\ref{THM:ExistsLSSHardness}, we
obtain the following.
\begin{theorem}
\label{THM:IsLSSCompleteness}
  $\IsLSS$ is $\coETR$-complete for 5-player games.
\end{theorem}

The remainder of this section is concerned with the proof of
Theorem~\ref{THM:ExistsLSSHardness}. We thus consider the polynomial
$F(y,z)$ as given in Definition~\ref{DEF:forallDHOMFOURFEASSimplex}
and assume as before that $F(y,z)\geq 0$ for all $y\in\zo^n$ and all
$z\in \RR^m$, and that each $F_i$ is symmetrized.

As explained above, to complete the proof of
Theorem~\ref{THM:ExistsLSSHardness} we just need to show that if there
exist $y \in \zo^n$ exists such that $F(y,z)>0$ for all
$z \in \Simplex^{m-1}$, then the game $\calG$ defined in the proof of
Theorem~\ref{THM:ExistsESSHardness} has an LSS. We thus assume that
$y \in \zo^n$ exists such that $F(y,z)>0$ for all
$z \in \Simplex^{m-1}$, and define $\sigma=\sigma_y$ as in
Lemma~\ref{LEM:Sigma-y}, with $y_0\in \zo$ arbitrarily chosen.

For $\eps>0$ to be specified later, consider a strategy
$\tau \neq \sigma$ such that $\norm{\sigma-\tau}_\infty<\eps$. We are
to prove that $u(\sigma;\tau^4) > u(\tau^5)$. Let us write
$\tau = \delta_1\tau_1+\delta_2\tau_2+\delta_3\tau_3$ as a convex
combination of strategies $\tau_j$ with
$\support(\tau_j)\subseteq S_j$, for $j=1,2,3$.  Since
$\support(\sigma) \subseteq S_1$ we have
\[
  \norm{\sigma-\tau}_\infty = \max\left(\norm{\sigma-\delta_1\tau_1}_\infty,\norm{\delta_2\tau_2}_\infty,\norm{\delta_3\tau_3}_\infty\right) \leq \eps \enspace .
\]
Since $\abs{S_2}=2$ and $\abs{S_3}=2m$ (recall that the actions of
$S_2 \cup S_3$ are duplicated), it follows that
$\norm{\tau_2}_\infty \geq \frac{1}{2}$ and
$\norm{\tau_3}_\infty \geq \frac{1}{2m}$. Combining this with the
inequalities $\norm{\delta_2\tau_2}_\infty \leq \eps$ and
$\norm{\delta_3\tau_3}_\infty \leq \eps$, it follows that
$\delta_2 \leq 2\eps$ and $\delta_3 \leq 2m\eps$. Since also
$\norm{\sigma-\delta_1\tau_1}_\infty \leq \eps$ we now have
\begin{equation}
\label{EQ:sigma-tau1-distance}
  \norm{\sigma-\tau_1}_\infty \leq \norm{\sigma-\delta_1\tau_1}_\infty + (1-\delta_1) \leq \eps + \delta_2 + \delta_3 \leq (2m+3)\eps \leq 4m\eps \enspace ,
\end{equation}
assuming, without loss of generality, $m\geq 2$ for the last
inequality.

It will be useful to introduce notation for the probabilities of the
strategy $\tau_1$.
\begin{definition}
  \label{DEF:tau1-probabilities}
  For a term of the form $t=(\alpha,i)$, define $b_t = y_i$. Next,
  let $p_{t,0}=\Pr_{\tau_1}[(t,b_t)]$,
$p_{t,1}=\Pr_{\tau_1}[(t,1-b_t)]$, and $p_t = p_{t,0}+p_{t,1}$.
\end{definition}
We also introduce notation for the set of actions in $S_1$ that are
inconsistent with $y$.
\begin{definition}
  \label{DEF:Bad-set-B}
  $B=\{(t,b) \in S_1 \mid b\neq b_t\}$.
\end{definition}

In order to prove that $u(\sigma;\tau^4) > u(\tau^5)$ we shall analyze
$u(\tau_j;\tau^4)$, for $j \in \{1,2,3\}$, separately, and compare to
$u(\sigma;\tau^4)$. Note that
$u(\sigma;\tau^4) = \delta_1^4 u(\sigma;\tau_1^4) + (1-\delta_1^4)T$.

\subsection{Comparison of $u(\sigma;\tau^4)$ to $u(\tau_1;\tau^4)$.}
Note first that
$u(\tau_1;\tau^4) = \delta_1^4 u(\tau_1;\tau_1^4) + (1-\delta_1^4)T$,
so it suffices to consider $u(\tau_1;\tau_1^4)$.

\begin{definition}
  For a fixed term $t$ and for $(t_j,b_j)$ chosen according to
  $\tau_1$ for $j=2,\dots,5$, we define events $A_t$, $B_{t,0}$, $B_{t,1}$ as
  follows. $A_t$ denotes the set of outcomes where
  $t \notin \{t_2,\dots,t_5\}$. $B_{t,0}$ denotes the set of outcomes
where $t \in \{t_2,\dots,t_5\}$, and $b_j=b_t$ whenever
$t_j=t$. Finally, $B_{t,1}$ denotes the set of outcomes where
$t \in \{t_2,\dots,t_5\}$, and $b_j=1-b_t$ whenever $t_j=t$.
\end{definition}
It is straightforward to  compute the probability of these events.
\begin{lemma}
\label{LEM:Events-At-Bt0-Bt1}
  For a fixed $t$ we have $\Pr_{\tau_1^4}[A_t] = (1-p_t)^4$,
  $\Pr_{\tau_1^4}[B_{t,0}] = (1-p_{t,1})^4-(1-p_t)^4$, and
  $\Pr_{\tau_1^4}[B_{t,1}] = (1-p_{t,0})^4-(1-p_t)^4$.
\end{lemma}
Using these we can express the payoffs of the two strategies $\sigma$ and
$\tau_1$ against $\tau_1^4$ in terms of the probabilities $p_{t,0}$ and $p_{t,1}$.
\begin{lemma}
  \label{LEM:payoffs-sigma-tau1}
  The payoffs $u(\sigma;\tau_1^4)$ and $u(\tau_1^5)$ satisfy
  the following equations.
\begin{align*}
  u(\sigma;\tau_1^4) &= \sum_t \frac{1}{M} \left((1-p_t)^4+(1-p_{t,1})^4\right)\\
  u(\tau_1^5) &= \sum_t p_t(1-p_t)^4 +  p_{t,0}(1-p_{t,1})^4 + p_{t,1}(1-p_{t,0})^4
\end{align*}
\end{lemma}
\begin{proof}
  By the definition of $u$ we have
\begin{align}
  u(\sigma;\tau_1^4) &= \sum_t \frac{1}{M} \left(2\Pr[A_t]+\Pr[B_{t,0}]\right)\\
  u(\tau_1^5) &= \sum_t 2 p_t \Pr[A_t] + p_{t,0}\Pr[B_{t,0}] + p_{t,1}\Pr[B_{t,1}]
\end{align}
The statement then follows from Lemma~\ref{LEM:Events-At-Bt0-Bt1}
\end{proof}
We start by relating the first terms in the two summations in
Lemma~\ref{LEM:payoffs-sigma-tau1}.
\begin{lemma}
  \label{LEM:tau1-bound1}

  We have the inequality
  \[
    \sum_t \frac{1}{M} (1-p_t)^4 \geq \left(1-\frac{1}{M}\right)^4 \enspace ,
  \]
  and when $p_t \leq \frac{2}{5}$ for all $t$ we furthermore have
  \[
    \left(1-\frac{1}{M}\right)^4  \geq \sum_t p_t (1-p_t)^4 \enspace .
  \]
  Both inequalities holds with equality if and only if
  $p_t=\frac{1}{M}$ for all~$t$.
\end{lemma}
\begin{proof}
  The function $(1-x)^4$ is strictly convex and thus Jensen's inequality gives
  \[
    \left(1-\frac{1}{M}\right)^4 = \left(1-\frac{\sum_t
        p_t}{M}\right)^4 \leq \frac{1}{M}\sum_t (1-p_t)^4 \enspace ,
  \]
  with equality if and only if $p_t=\frac{1}{M}$ for all~$t$,
  resulting in the first of the stated inequalities. Next, since
  $\frac{d^2}{dx} x(1-x)^4 = 4(1-x)^2(5x-2)$, the function $x(1-x)^4$
  is strictly concave in the interval $[0,\frac{2}{5}]$ and thus
  Jensen's inequality gives
  \[
    \frac{1}{M} \sum_t p_t (1-p_t)^4 \leq \left(\frac{\sum_t p_t}{M}\right)\left(1-\frac{\sum_t p_t}{M}\right)^4 = \frac{1}{M}\left(1-\frac{1}{M}\right)^4 \enspace ,
  \]
  with equality if and only if $p_t=\frac{1}{M}$ for all~$t$,
  resulting in the second of the stated inequalities
\end{proof}
Next we relate the remaining terms in the two summations in
Lemma~\ref{LEM:payoffs-sigma-tau1}.
\begin{lemma}
  \label{LEM:tau1-bound2}
  Suppose that $p_{t,0} \geq \frac{5}{6M}$ for all $t$. Then
  \[
\sum_t \frac{1}{M} (1-p_{t,1})^4 \geq \sum_t p_{t,0}(1-p_{t,1})^4 + p_{t,1}(1-p_{t,0})^4 +  \frac{1}{M} p_{t,1} \enspace .
  \]
\end{lemma}
\begin{proof}
  It is easy to verify that for all $0\leq x \leq \frac{1}{6}$ it
  holds that $1-4x \leq (1-x)^4 \leq 1-3x$.  Using this we get
  \[
    \begin{split}
    \sum_t & \frac{1}{M} (1-p_{t,1})^4 - p_{t,0}(1-p_{t,1})^4 -
    p_{t,1}(1-p_{t,0})^4  \\ \geq \sum_t & \frac{1}{M}(1-4p_{t,1}) - p_{t,0}(1-3p_{p,1}) - p_{p,1}(1-3p_{t,0}) \\
    = \sum_t & p_{t,1}\left(6p_{t,0} - \frac{4}{M}\right) \\
    \geq  \sum_t & \frac{1}{M}p_{t,1} \enspace .
  \end{split}
\]
\end{proof}
Combining these we obtain the following.
\begin{proposition}
\label{PROP:analysis-tau1}
  Assume that $\norm{\sigma-\tau_1}_\infty \leq \frac{1}{6M}$. Then
  \[
    u(\sigma;\tau_1^4) \geq u(\tau_1^5) + \frac{1}{M}\Pr_{\tau_1}[B] \enspace .
  \]
  Furthermore, $u(\sigma;\tau_1^4)=u(\tau_1^5)$ if and only if
  $\sigma=\tau_1$.
\end{proposition}
\begin{proof}
  The inequality follows by the equations for $u(\sigma;\tau_1^4)$ and
  $u(\tau_1^5)$ given in Lemma~\ref{LEM:payoffs-sigma-tau1} together
  with the inequalities of Lemma~\ref{LEM:tau1-bound1} and
  Lemma~\ref{LEM:tau1-bound2}. If $u(\sigma;\tau_1^4)=u(\tau_1^5)$ it
  follows that $\Pr_{\tau_1}[B]=0$ in which case
  Lemma~\ref{LEM:tau1-bound1} expresses the inequality
  $u(\sigma;\tau_1^4) \geq u(\tau_1^5)$ and states that it holds with
  equality if and only if $\sigma=\tau_1$.
\end{proof}
This finally allows us to compare playing the strategies $\sigma$ and
$\tau_1$ against $\tau^4$.
\begin{corollary}
\label{COR:analysis-tau1}
  Assume that $\norm{\sigma-\tau_1}_\infty \leq \frac{1}{6M}$. Then
  \[
    u(\sigma;\tau^4) \geq u(\tau_1;\tau^4) + \frac{\delta_1^4}{M}\Pr_{\tau_1}[B] \enspace .
  \]
  Furthermore, $u(\sigma;\tau^4)=u(\tau_1;\tau^4)$ if and only if
  $\sigma=\tau_1$.
\end{corollary}
\begin{proof}
  The inequality follows from Proposition~\ref{PROP:analysis-tau1}
  together with the observations that
  $u(\sigma;\tau^4) = \delta_1^4 u(\sigma;\tau_1^4) + (1-\delta_1^4)T$
  and
  $u(\tau_1;\tau^4) = \delta_1^4 u(\tau_1;\tau_1^4) +
  (1-\delta_1^4)T$.
\end{proof}

\subsection{Comparison of $u(\sigma;\tau^4)$ to $u(\tau_2;\tau^4)$ and  $u(\tau_3;\tau^4)$.}
Let in the following $\widehat{\tau_1}$ denote the strategy obtained
from $\tau$ by conditioning on the outcome belonging to
$S_1 \setminus B$. We first consider playing
$\sigma$ against $\widehat{\tau_1}$.
\begin{lemma}
  \label{LEM:sigma-best-when-consistent}
  $u(\sigma;\widehat{\tau_1}^4) \geq T$.
\end{lemma}
\begin{proof}
  Similarly to Definition~\ref{DEF:tau1-probabilities}, for a term $t$
  we let $\widehat{p}_{t,0}=\Pr_{\widehat{\tau_1}}[(t,b_t)]$,
  $\widehat{p}_{t,1}=\Pr_{\widehat{\tau_1}}[(t,1-b_t)]$, and
  $\widehat{p}_t = \widehat{p}_{t,0}+\widehat{p}_{t,1}$. Clearly
  $\widehat{p}_{t,1}=0$ and $\widehat{p}_t=\widehat{p}_{t,0}$.  Using
  $\widehat{\tau_1}$ in place of $\tau_1$ in
  Lemma~\ref{LEM:payoffs-sigma-tau1} and Lemma~\ref{LEM:tau1-bound1}
  we then have
  \[
    \begin{split}
      u(\sigma;\widehat{\tau_1}^4) & = \sum_t \frac{1}{M} \left((1-\widehat{p}_t)^4+(1-\widehat{p}_{t,1})^4\right) \\
      & = 1 + \sum_t \frac{1}{M} (1-\widehat{p}_t)^4 \geq 1 + \left(1-\frac{1}{M}\right)^4 = T \enspace .
  \end{split}
\]
\end{proof}
We first compare playing the strategies $\sigma$ and $\tau_2$ against
$\tau^4$.
\begin{proposition}
\label{PROP:analysis-tau2}
  We have
  \[
    u(\sigma;\tau^4) \geq u(\tau_2;\tau^4) - 12\delta_1\Pr_{\tau_1}[B] + \delta_2 \enspace .
  \]
\end{proposition}
\begin{proof}
  We divide the outcomes of $\tau^4$ into events. Let $C_1$ be the
  event that the four players play an action in $S_1 \setminus B$,
  i.e.\ $C_1=(S_1\setminus B)^4$. Let $C_2$ be the event that at least
  one of the four players play an action in $S_2$. Let $C_3$ be the
  event that none of the four players play on action in $S_2$ and that
  at least one player is playing an action of $S_3$. Finally let $C_4$
  be the event consisting of all remaining outcomes.

  Conditioned on the event $C_1$, the strategy $\sigma$ receive payoff
  at least~$T$ against $\tau^4$ by
  Lemma~\ref{LEM:sigma-best-when-consistent}, whereas $\tau_2$
  receives payoff~$T$.  When $C_2$ occurs, which happens with
  probability at least $\delta_2$, the strategy $\sigma$ receives
  payoff~$T$ while the strategy $\tau_2$ receives payoff~$0$. When
  $C_3$ occurs, both strategies $\sigma$ and $\tau_2$ receives
  payoff~$T$. We finally consider the case of $C_4$ occurring. Note
  that for this to happen, at least one of the four players need to
  play an action of $B$, which in turn happens with probability at
  most $4\delta_1\Pr_{\tau_1}[B]$. By definition of $u$, the strategy
  $\tau_2$ receives at most payoff~$T+1$ while $\sigma$ receives at
  least $0$ in any event and in particular in the event
  $C_4$. Combining these observations we obtain
  \[
    u(\sigma;\tau^4) - u(\tau_2;\tau^4) \geq \delta_2T - 4\delta_1\Pr_{\tau_1}[B](T+1) \enspace ,
  \]
  from which the stated inequality follows by using also that
  $1 \leq T \leq 2$.
\end{proof}
We next compare playing the strategies
$\sigma$ and $\tau_3$ against $\tau^4$.
\begin{proposition}
  \label{PROP:analysis-tau3}
  We have
\[
  u(\sigma; \tau^4) \geq  u(\tau_3;\tau^4) + 4\delta_1\delta_3^3\left(T - u(\tau_3; \tau_3^3, \tau_1)\right) - 8\delta_1\Pr_{\tau_1}[B] \enspace .
\]
\end{proposition}
\begin{proof}
  We first consider playing $\sigma$ against $\tau^4$. In the event that
the four players play an action in $S_1 \setminus B$, the strategy
$\sigma$ receive payoff at least~$T$ by
Lemma~\ref{LEM:sigma-best-when-consistent}. When at least one player
plays an action outside $S_1$, the strategy $\sigma$ receives $T$. It
follows that for $\sigma$ to receive payoff less than $T$ at least one
of the four players need to play an action of $B$ which happens with
probability at most $4\delta_1\Pr_{\tau_1}[B]$. We thus have
\[
  u(\sigma; \tau^4) \geq (1-4\delta_1\Pr_{\tau_1}[B])T = T - 4\delta_1\Pr_{\tau_1}[B]T\enspace .
\]
We next consider playing $\tau_3$ against $\tau^4$. The strategy
$\tau_3$ receives payoff~$T$ unless exactly one player plays an action
of $S_1$ and the remaining three players play an action of $S_3$. We
thus have
\[
  u(\tau_3;\tau^4) = 4\delta_1\delta_3^3 u(\tau_3; \tau_3^3, \tau_1) +
  (1-4\delta_1\delta_3^3)T = T - 4\delta_1\delta_3^3(T - u(\tau_3; \tau_3^3, \tau_1)) \enspace .
\]
By combining these, and also using $T\leq 2$, the stated inequality
follows.
\end{proof}

\subsection{Comparison of $u(\sigma;\tau^4)$ to $u(\tau^5)$.}
We now combine the analysis of the previous subsections and make use
of the assumption about $F(y,z)>0$ for all $z \in \Simplex^{m-1}$. We
make use of the latter assumption in connection with the relationship
between $\calG$ and $F$ given by Lemma~\ref{LEM:Relationship-G-F}.
\begin{lemma}
  \label{LEM:Utility-in-neighborhood}
  There exists $\eps_y>0$ such that when $\tau_1 \in \Delta(S_1)$ is
  such that $\norm{\sigma-\tau_1}_\infty \leq \eps_y$ we have
  $u(\tau_3; \tau_3^3,\tau_1)<T$ for all $\tau_3 \in \Delta(S_3)$.
\end{lemma}
\begin{proof}
  By assumption we have $F(y,z)>0$ for all $z \in
  \Simplex^{m-1}$. Since $F$ is continuous and $\Simplex^{m-1}$
  compact this implies that
  \begin{equation}
    \label{EQ:min-F-positive}
    \min_{z \in \Simplex^{m-1}} F(y,z) > 0 \enspace .
  \end{equation}
Next, define the function $G : \Delta(S_1) \rightarrow \RR$ by
\[
  G(\tau_1)=\max_{\tau_3 \in \Delta(S_3)} u(\tau_3; \tau_3^3, \tau_1)
\]
Lemma~\ref{LEM:Relationship-G-F} and
Equation~(\ref{EQ:min-F-positive}) together implies that
$G(\sigma) < T$. Since the function $u$ is continuous and again using
compactness of $\Simplex^{m-1}$ it follows that $G$ is continuous as
well. This means there exists $\eps_y>0$ such that $G(\tau_1)<T$
whenever $\norm{\sigma-\tau_1}_\infty < \eps_y$.
\end{proof}

We now define
$\eps=\frac{1}{4m}\min\left(\eps_y,\frac{1}{11M}\right)$. This means
by Equation~(\ref{EQ:sigma-tau1-distance}) that in particular
$\norm{\sigma-\tau_1}_\infty \leq \frac{1}{6M}$, satisfying the
condition of Corollary~\ref{COR:analysis-tau1}, and
$\norm{\sigma-\tau_1}_\infty \leq \eps_y$, satisfying the condition of
Lemma~\ref{LEM:Utility-in-neighborhood}

Combining the inequalities of Corollary~\ref{COR:analysis-tau1},
Proposition~\ref{PROP:analysis-tau2}, and
Proposition~~\ref{PROP:analysis-tau3}, and using that
$\tau=\delta_1\tau_1+\delta_2\tau_2+\delta_3\tau_3$ results in the
following inequality.
\begin{equation}
  \begin{split}
    u(\sigma;\tau^4) \geq u(\tau^5) &+ \left(\frac{\delta_1^5}{M}-12\delta_1\delta_2-8\delta_1\delta_3\right)\Pr_{\tau_1}[B] \\ &+ \delta_2^2 + 4\delta_1\delta_3^4\left(T - u(\tau_3; \tau_3^3, \tau_1)\right)
  \end{split}
\end{equation}

The definition of $\eps$ also gives the inequality
$\eps \leq \frac{1}{44mM}$. Note that
$M(12\delta_2 + 8\delta_3) \leq M(24\eps + 16m\eps) \leq 32mM\eps$. On
the other hand
$\delta_1^4 \geq (1-\delta_2-\delta_3)^4 \geq (1-3m\eps)^4 > 1-12m\eps
\geq 1-12mM\eps$. This gives that
$(\frac{\delta_1^5}{M}-12\delta_1\delta_2-8\delta_1\delta_3) >
\frac{\delta_1}{M}(1-44mM\eps) \geq 0$.

From this and Lemma~\ref{LEM:Utility-in-neighborhood} it follows that
$u(\sigma;\tau^4) \geq u(\tau^5)$. Furthermore, an equality,
$u(\sigma;\tau^4) = u(\tau^5)$ implies that $\delta_2=0$ and
$\delta_3=0$. This means that $\tau=\tau_1$, and
Corollary~\ref{COR:analysis-tau1} then implies that in fact
$\sigma=\tau$. This concludes the proof of
Theorem~\ref{THM:ExistsLSSHardness}.

\section{Conclusion}
We have shown the problems $\ExistsESS$ and $\ExistsLSS$ to be hard
for $\existsDOP \coETR$ and members of $\cExistsForallR$. The main
open problem is to characterize the precise complexity of $\ExistsESS$
and $\ExistsLSS$, perhaps by improving the upper bounds. Another point
is that our hardness proofs construct 5-player games, whereas the
recent and related $\cETR$-completeness results for decision problems
about NE in multi-player games holds already for~3-player games. This
leads to the question about the complexity of $\ExistsESS$ and
$\IsESS$ as well as $\ExistsLSS$ and $\IsLSS$ in 3-player and 4-player
games. The reason that we end up with 5-player games is that we
construct a degree~4 polynomial in the reduction, rather than (a
system of) degree~2 polynomials as used in the related
$\cETR$-completeness results. In both cases a number of players equal
to the degree is used to simulate evaluation of a monomial and a last
player is used to select the monomial. For our proof we critically use
that the degree~4 polynomial involved in the reduction may be assumed
to be non-negative.

\printbibliography

\end{document}